\documentclass[12pt, twoside]{article}

\usepackage{amsmath,amsthm,amssymb}

\usepackage{times}

\usepackage{enumerate}

 \usepackage{amsmath}
\usepackage{amssymb}
\usepackage{amsthm}
\usepackage{graphicx}
\usepackage{amscd}
\usepackage{graphicx}

\usepackage{amssymb}
\usepackage{amsmath}
\usepackage{latexsym}
\usepackage[numbers,sort&compress]{natbib}
\usepackage{verbatim}
\usepackage{physics}
\usepackage{dsfont}	
\usepackage{microtype}  
\usepackage{hyperref}
\usepackage{tikz}
\usepackage{booktabs}   

\usepackage{color}

\date{}

\newtheorem{theorem}{Theorem}
\newtheorem{lemma}[theorem]{Lemma}

\newtheorem{example}[theorem]{Example}

\title{The equivalence of linear codes implies semi-linear equivalence}

\begin{document}
\author{Simeon Ball and James Dixon\thanks{15 July 2021.}} 
\maketitle

\begin{abstract}
We prove that if two linear codes are equivalent then they are semi-linearly equivalent. We also prove that if two additive MDS codes over a field are equivalent then they are additively equivalent.
\end{abstract}

\section{Introduction}

Let $F$ be a finite set. A {\em code} $\mathcal A$ of {\em length} $n$ is a subset of $F^n$. The Hamming distance $d(u,v)$ between $u,v \in F^n$ is the number of coordinates in which $u$ and $v$ differ. As in \cite{KDO2015}, we say that two codes $\mathcal A$ and $\mathcal B$ are {\em equivalent} if, after some permutation $\alpha$ of the coordinates of the elements of $\mathcal A$, there exist permutations $\sigma_i$ of $F$ such that for all $u=(u_1,\ldots,u_n) \in \mathcal A$, 
$$
(\sigma_1(u_1),\ldots,\sigma_n(u_n)) \in \mathcal B.
$$
Let $\psi$ denote the bijection from $\mathcal A$ to $\mathcal B$ induced by these permutations. Then it is clear that $\psi$ preserves Hamming distance; that is
$$
d(u,v)=d(\psi(u),\psi(v))
$$
for all $u,v \in \mathcal A$.


If $F$ is a finite field and $\mathcal A$ and $\mathcal B$ are $k$-dimensional subspaces of $F^n$ then we say that $\mathcal A$ and $\mathcal B$ are {\em linear codes}. Two codes are $\mathcal A$ and $\mathcal B$ are {\em linearly equivalent} in the above definition of equivalence if, for all $i \in \{1,\ldots,n\}$, 
$$
\sigma_i(x)=\lambda_i x,
$$
for some non-zero $\lambda_i \in F$. In the case that $F$ has non-trivial automorphisms (fields of non-prime order), for a fixed automorphism $\beta$ of $F$, applying the permutation
$$
\sigma_i(x)=\lambda_i x^{\beta},
$$
to the $i$-th coordinate of the codewords of $\mathcal A$, we get a linear code $\mathcal B$ which is equivalent to $\mathcal A$. 
This more general equivalence we will call {\em semi-linear equivalence}, following \cite{BBFKKW2006}, \cite{Fripertinger2005} and \cite{LW2011}. It is also called {\em P$\Gamma$L-equivalence} \cite{BGL2021} and is referred to as simply {\em equivalence} in \cite{Bouyukliev2007} and \cite{HP2003}.


The main result of this note is the following theorem.

\begin{theorem} \label{mainthm}
Two linear codes are equivalent if and only if they are semi-linearly equivalent.
\end{theorem}
The reverse implication is trivial based on the above discussion, so the goal here is to prove that if $\mathcal A$ and $\mathcal B$ are equivalent linear codes then they are semi-linearly equivalent.

Theorem~\ref{mainthm} can be compared to Theorem 1.5.10 in \cite{BBFKKW2006}. In that theorem they prove that if there is a Hamming distance preserving bijection of $F^n$ which maps subspaces to subspaces then this map is semi-linear, i.e. it is additive and $\sigma(\lambda x)= \lambda^{\beta} \sigma(x)$, for all $\lambda \in F$, where $\beta$ is an automorphism of $F$. 

Theorem~\ref{mainthm} has significant implications when classifying codes. For example, in a classification of codes up to equivalence, as in \cite{Alderson2006}, \cite{KDO2015} and \cite{KO2016}, one will not find two linear codes in the same equivalence class. Perhaps of more interest, in a classification of linear codes up to semi-linear equivalence, as in \cite{BGL2021} and \cite{BBK2020}, we can now be sure that two semi-linearly inequivalent codes are also inequivalent.

\section{Equivalence implies semi-linear equivalence.}

Let $\mathcal A$ and $\mathcal B$ be $k$-dimensional linear codes over a finite field ${\mathbb F}_q$ of length $n$ which are equivalent. Since $\alpha(\mathcal A)$ is also a linear code equivalent to $\mathcal B$, we can assume that $\alpha$, the permutation of the coordinates of the elements of $\mathcal A$, is the identity by replacing $\mathcal A$ by $\alpha(\mathcal A)$. Indeed, by re-ordering the coordinates in both codes, if necessary, we can find generator matrices $\mathrm A$ and $\mathrm B$ in standard form for $\mathcal A$ and $\mathcal B$ respectively. Recall that a generator matrix in standard form is a $k\times n$ matrix in which the initial $k \times k$ sub-matrix is the identity matrix.

As above, we denote by $\sigma_i$ the permutation in the $i$-th coordinate which has the property that for all $u=(u_1,\ldots,u_n) \in \mathcal A$, 
$$
(\sigma_1(u_1),\ldots,\sigma_n(u_n)) \in \mathcal B.
$$
 
 \begin{lemma} \label{additivelemma}
 If not all of the columns of $A$ and $B$ are of weight one then, for each $j \in \{1,\ldots,n\}$, the map
 $$
 \tau_j(x)=\sigma_j(x)-\sigma_j(0)
 $$
 is additive.
 \end{lemma}
 
 \begin{proof}
 Let $A=(\alpha_{jr})$ and let $B=(\beta_{jr})$.
 
We have that for each $a=(a_1,\ldots,a_k) \in {\mathbb F}_q^k$, there is a $b=(b_1,\ldots,b_k) \in {\mathbb F}_q^k$, such that
\begin{equation} \label{eq0}
 (\sigma_1(a_1),\ldots,\sigma_k(a_k),\sigma_{k+1}(\sum_{j=1}^k \alpha_{j,k+1}a_j),\ldots)
 =
 (b_1,\ldots,b_k,\sum_{j=1}^k \beta_{j,k+1} b_j,\ldots).
 \end{equation}
 
Hence, for $j \in \{1,\ldots,k\}$,
 $$
 b_j=\sigma_j(a_j)
 $$
 and
\begin{equation} \label{eq1}
 \sum_{j=1}^k \beta_{jr} \sigma_j(a_j)=\sigma_{r}( \sum_{j=1}^k \alpha_{jr} a_j),
 \end{equation}
 for $r \in \{ k+1,\ldots,n\}$.
 
 Setting $a_i=0$ for $i\neq j$ we deduce that
\begin{equation} \label{eq1.5}
 \beta_{jr}\sigma_j(a_j)+\sum_{i \neq j} \beta_{ir} \sigma_i(0) =\sigma_{r}( \alpha_{jr} a_j).
 \end{equation}
Summing over $j$, and substituting in (\ref{eq1})
 $$
 \sum_{j=1}^k \sigma_{r}( \alpha_{jr} a_j)- \sum_{j=1}^k\sum_{i \neq j} \beta_{ir} \sigma_i(0)=\sigma_{r}( \sum_{j=1}^k \alpha_{jr} a_j)
 $$
 which implies
  $$
 \sum_{j=1}^k \sigma_{r}( \alpha_{jr} a_j)-(k-1)\sum_{i=1}^k \beta_{ir} \sigma_i(0)=\sigma_{r}( \sum_{j=1}^k \alpha_{jr} a_j).
 $$
Again from (\ref{eq1})
\begin{equation} \label{eq3}
 \sum_{i=1}^k \beta_{ir} \sigma_i(0)=\sigma_{r}(0).
\end{equation}
Thus,
\begin{equation}  \label{eq3.5}
 \sum_{j=1}^k \sigma_{r}( \alpha_{jr} a_j)- (k-1) \sigma_{r}(0)=\sigma_{r}( \sum_{j=1}^k \alpha_{jr} a_j).
\end{equation}
 
Define
$$
\tau_{r}(x)=\sigma_{r}(x)-\sigma_{r}(0).
$$ 
Then (\ref{eq3.5}) becomes
\begin{equation} \label{eq4}
 \sum_{j=1}^k \tau_{r}( \alpha_{jr} a_j)=\tau_{r}( \sum_{j=1}^k \alpha_{jr} a_j).
\end{equation}

 Firstly, assume that the $r$-th column of $A$ is not of weight one. Without loss of generality $\alpha_{1r}, \alpha_{2r} \neq 0$, so substituting in (\ref{eq4}), $a_1=\alpha_{1r}^{-1}x$ and $a_2=\alpha_{2r}^{-1}y$, $a_j=0$ for $j \in \{3,\ldots,k\}$, we have that
$$
\tau_{r}(x+y)=\tau_{r}(x)+\tau_{r}(y).
$$
We conclude that $\tau_{r}$ is an additive permutation of ${\mathbb F}_q$.

Subtracting (\ref{eq3}) from (\ref{eq1.5}), we have that 
$$
\beta_{jr}( \sigma_j(a_j)-\sigma_j(0))=\tau_{r}(\alpha_{1r} a_1)
$$
This implies that 
$$
\tau_j(x)=\sigma_j(x)-\sigma_j(0)
$$
is additive too, for $j \in \{1,\ldots,k\}$.


If the $r$-th column of $A$ is of weight one then, from (\ref{eq1}),
$$
\sigma_r(\alpha_{ir} a_i)=  \sum_{j=1}^k \beta_{jr} \sigma_j(a_j)
$$
for some $i \in\{1,\ldots,k\}$. Since this holds for all $(a_1,\ldots,a_k) \in {\mathbb F}_q^k$, we conclude that $\beta_{jr}=0$ for $j\neq i$ and hence,
$$
\sigma_r(\alpha_{ir} a_i)=  \beta_{ir} \sigma_i(a_i).
$$
This implies $\sigma_r(0)=  \beta_{ir} \sigma_i(0)$ and so 
$$
\tau_r(\alpha_{ir} a_i)=  \beta_{ir} \tau_i(a_i).
$$
Since $\tau_i$ is additive, this implies $\tau_r$ is additive.
\end{proof}

The following lemma implies that we can assume that $\sigma_i$ is additive, by replacing $\sigma_i$ by $\tau_i$, for each $i \in \{1,\ldots,n\}$.

\begin{lemma} \label{itsadditive}
 If not all of the columns of $A$ and $B$ are of weight one then, for all $u=(u_1,\ldots,u_n) \in \mathcal A$, 
$$
(\tau_1(u_1),\ldots,\tau_n(u_n)) \in \mathcal B.
$$ 
\end{lemma}

\begin{proof}

By Lemma~\ref{additivelemma} and (\ref{eq0}), for each $a=(a_1,\ldots,a_k) \in {\mathbb F}_q^k$, there is a $b=(b_1,\ldots,b_k) \in {\mathbb F}_q^k$, such that
$$
 (\tau_1(a_1),\ldots,\tau_k(a_k),\tau_{k+1}(\sum_{j=1}^k\alpha_{j,k+1}a_j),\ldots)+v=(b_1,\ldots,b_k)B,
$$
where $v=(v_1,\ldots,v_n) \in {\mathbb F}_q^n$ is defined by $v_j=\sigma_j(0)$.

Since $\tau_i$ is additive, $\tau_i(0)=0$. Thus, with $a=(0,\ldots,0)$ the above implies there is a $(b_1',\ldots,b_k')$ such that
$$
v=(b_1',\ldots,b_k')B.
$$
Thus
$$
 (\tau_1(a_1),\ldots,\tau_k(a_k),\tau_{k+1}(\sum_{j=1}^k\alpha_{j,k+1}a_j),\ldots)=((b_1,\ldots,b_k)-(b'_1,\ldots,b'_k))B,
$$
and hence for each $a=(a_1,\ldots,a_k) \in {\mathbb F}_q^k$, there is a $c=(c_1,\ldots,c_k) \in {\mathbb F}_q^k$, such that
$$
 (\tau_1(a_1),\ldots,\tau_k(a_k),\tau_{k+1}(\sum_{j=1}^k\alpha_{j,k+1}a_j),\ldots)=(c_1,\ldots,c_k)B.
$$
\end{proof}

We are now in a position to prove the main theorem. Note that an additive permutation of ${\mathbb F}_{p^h}$ is of the form
$$
x \mapsto \sum_{i=0}^{h-1} c_i x^{p^i}.
$$
This is easily verified since if we consider ${\mathbb F}_{p^h}$ as ${\mathbb F}_{p}^h$ then an additive map is given by an $h \times h$ matrix over ${\mathbb F}_{p}$. Thus, there are $p^{h^2}$ in total, which coincides with the number of functions which can be defined as above.

\begin{proof} (of Theorem~\ref{mainthm})

Suppose all of the columns of $A$ and $B$ are of weight one. Then, after a suitable permutation of the coordinates, $\mathcal A$ is linearly equivalent to a code $\mathcal A'$ whose codewords are
 \begin{equation} \label{eq13}
 (\underbrace{a_1,\ldots,a_1}_{m_1},\underbrace{a_2,\ldots,a_2}_{m_2}, \ldots,\underbrace{a_k,\ldots,a_k}_{m_k}).
\end{equation}
 Thus, since $\mathcal A'$ and $\mathcal B$ are equivalent, there are permutations $\theta_i$ and a function $f$ from $\{1,\ldots,n\}$ to $\{1,\ldots,k\}$, such that for all $(b_1,\ldots,b_k) \in {\mathbb F}_q^k$,
 $$
 (\theta_1(b_{f(1)} \beta_{f(1),1}),\theta_2(b_{f(2)} \beta_{f(2),2}),\ldots,\theta_n(b_{f(n)} \beta_{f(n),n}))
 $$
is equal to (\ref{eq13}). We can scale the columns of $B$ and obtain a code $\mathcal B'$ which is linearly equivalent to a code $\mathcal B$ and has codewords
$$
(\theta_1(b_{f(1)}),\theta_2(b_{f(2)} ),\ldots,\theta_n(b_{f(n)})).
$$
Now, comparing with (\ref{eq13}) we have that, for example, if $m_1 \geqslant 2$, 
$$
\theta_1(b_{f(1)})=\theta_2(b_{f(2)})
$$
for all $(b_1,\ldots,b_k) \in {\mathbb F}_q^k$. This implies $f(1)=f(2)$, and hence $\theta_1=\theta_2$, etc. Thus, we conclude that $\mathcal B'$ has codewords
 \begin{equation} 
 (\underbrace{b_{g(1)},\ldots,b_{g(1)}}_{m_1},\underbrace{b_{g(2)},\ldots,b_{g(2)}}_{m_2}, \ldots,\underbrace{b_{g(k)},\ldots,b_{g(k)}}_{m_k})
\end{equation}
and is linearly equivalent to $\mathcal A'$. Hence, $\mathcal A$ is linearly equivalent to $\mathcal B$.

Suppose not all of the columns of $A$ and $B$ are of weight one.

By Lemma~\ref{itsadditive}, we can assume for $j \in \{1,\ldots,n\}$,
$$
\sigma_j(x)=\sum_{i} c_{ji} x^{p^{i}}.
$$
As in the previous proofs, let $B=(\beta_{jr})$ and let $A=(\alpha_{jr})$.

Substituting in (\ref{eq1}), we have that for all $(a_1,\ldots,a_k) \in {\mathbb F}_q^k$ and $r \in \{k+1,\ldots,n\}$, 
$$
 \sum_{j=1}^k \sum_{i=0}^{h-1} \beta_{jr} c_{ji} a_j^{p^{i}}=\sum_{i} c_{ri} (  \sum_{j=1}^k \alpha_{jr} a_j)^{p^{i}}.
$$
Therefore, for all $j \in \{1,\dots,k\}$, $r \in \{k+1,\ldots,n\}$ and $i \in \{0,\ldots,h-1\}$,
\begin{equation} \label{bigeqn}
 \beta_{jr} c_{ji}=c_{ri}\alpha_{jr} ^{p^{i}}.
\end{equation}
Let $t\in \{0,\ldots,h-1\}$ be minimal such that $c_{1t} \neq 0$. We aim to show that we can multiply columns and rows of $B$ by non-zero elements of ${\mathbb F}_q$ in such a way that the matrix $B$ will transform into $A^{p^{t}}$. If both $\alpha_{jr}$ and $\beta_{jr}$ are zero then they will be unaffected by these multiplications and so for this particular entry the previous statement holds. Thus, we assume that  $\alpha_{jr}$ and $\beta_{jr}$ are not both zero for any particular $j \in \{1,\dots,k\}$ and $r \in \{k+1,\ldots,n\}$.

If $\beta_{jr}=0$ then, by the previous paragraph, we can assume $\alpha_{jr}\neq 0$. Hence, $c_{ri}=0$ for all $i \in \{0,\ldots,h-1\}$ which implies $\sigma_r=0$, a contradiction.

If $\alpha_{jr}=0$ then, by the previous paragraph, we can assume $\beta_{jr}\neq 0$. Hence, $c_{ji}=0$ for all $i \in \{0,\ldots,h-1\}$ which implies $\sigma_j=0$, a contradiction.


Thus, we can assume $c_{1t}$, $\alpha_{1r}$ and $\beta_{1r}$ are non zero.  Thus, (\ref{bigeqn}) implies $c_{rt} \neq 0$ for all $r \in \{k+1,\ldots,n\}$, and so 
$$
\frac{\beta_{jr}c_{jt}}{\beta_{1r}c_{1t}}=\frac{\alpha_{jr}^{p^t}}{\alpha_{1r}^{p^t}}.
$$
If $c_{jt}=0$ for some $j$ then $\alpha_{jr}=0$, which we have already ruled out.

Thus, we can divide the $r$-th column of $B$ by $\beta_{1r}$ and multiply the $j$-th row of this generator matrix by 
$$
\frac{c_{jt}}{c_{1t}}
$$ 
(which does not change the code) and we get a generator matrix whose $j$-th coordinate of the $r$-th column is
 $$
\frac{\alpha_{jr}^{p^t}}{\alpha_{1r}^{p^t}}
$$

Now, multiplying  the $r$-th column of the generator matrix by $\alpha_{1r}^{p^t}$, we get a generator matrix of a code linearly equivalent to $\mathcal B$ whose $r$-th column is the $r$-th column of $A^{p^{t}}$. 

Therefore, we conclude that $\mathcal A$ and $\mathcal B$ are semi-linearly equivalent.
\end{proof}

\section{Additive MDS codes}

The motivation for this article stems from the discussion after Theorem 3.4 in \cite{BGL2021}. In that discussion the following question is asked. ``If two additive codes $\mathcal A$ and $\mathcal B$ are equivalent then are they necessarily P$\Gamma$L-equivalent? Equivalently, if two additive codes $\mathcal A$ and $\mathcal B$ are P$\Gamma$L-inequivalent then is it true that they are inequivalent?" They also pose the same question for linear codes.

We have proved in Theorem~\ref{mainthm} that it is true for linear codes. Here, we answer the question in the affirmative for additive MDS codes. Recall that an MDS code is a code attaining the Singleton bound. That is, if $C$ is a MDS code of length $n$ over $F$ of minimum distance $d$ then $|C|=|F|^{n-d+1}$.

An additive code over ${\mathbb F}_{q}$ is linear over the prime subfield ${\mathbb F}_p$, where $q=p^h$. Two additive codes $\mathcal A$ and $\mathcal B$ are {\em additively equivalent} if, after some permutation of the coordinates of the elements of $\mathcal A$, there exist permutations 
$$
\sigma_j =\sum_{i=0}^{h-1} c_{ij} x^{p^i}.
$$
such that for all $u=(u_1,\ldots,u_n) \in \mathcal A$, 
$$
(\sigma_1(u_1),\ldots,\sigma_n(u_n)) \in \mathcal B.
$$

An additive MDS code of size $q^k$ is the row space over ${\mathbb F}_p$ of a $kh \times n$ matrix, whose elements are from ${\mathbb F}_q$. We can consider the elements of ${\mathbb F}_{q}$ as row vectors of ${\mathbb F}_p^h$, so that the generator matrix becomes a $kh \times nh$ matrix. 

It is fairly easy to show that we can use Gaussian elimination and additive permutations to obtain a generator matrix for an additive MDS code, additively equivalent to the original code, whose generator matrix is of the form
$$
\left(
\begin{array}{ccccccccccccc}
I_h & O_h & O_h & O_h &  \ldots & \ldots & O_h & O_h & \alpha_{r,1} & \alpha_{r+1,1} & \ldots \\
O_h & I_h & O_h & O_h &  \ldots & \ldots & O_h & O_h & \alpha_{r,2} & \alpha_{r+1,2} &  \ldots \\
O_h & O_h & I_h & O_h &  \ldots &  \ldots & O_h & O_h & \alpha_{r,3} & \alpha_{r+1,3}  & \ldots \\
O_h & O_h & O_h & I_h &  \ddots &  \ldots & O_h & O_h &\alpha_{r,4} & \alpha_{r+1,4}  & \ldots \\
\vdots & \vdots &  \ldots &  \ddots & \ddots & \ddots & \ddots & \vdots & \vdots& \vdots \\
\vdots & \vdots &  \ldots &  \ldots & \ldots & \ddots & \ddots & \vdots & \vdots & \vdots\\
O_h & O_h & O_h & O_h &  \ldots &  \ldots & I_h & O_h &\alpha_{r,k-1} & \alpha_{r+1,k-1}  &  \ldots \\
O_h & O_h & O_h & O_h & \ldots &  \ldots &  O_h & I_h &\alpha_{r,k} & \alpha_{r+1,k}  &  \ldots \\
\end{array}
\right),
$$
where $I_h$ is the $h\times h$ identity matrix, $O_h$ is the $h \times h$ zero matrix and $\alpha_{i,j}$ are $h\times h$ non-singular matrices. See \cite{BGL2021}, for more details on this.

Suppose that the additive MDS codes $\mathcal A$ and $\mathcal B$ are equivalent. We can mimic the proof of Lemma~\ref{additivelemma} and Lemma~\ref{itsadditive}, replacing $a_i$ and $b_i$ by row vectors in ${\mathbb F}_p^h$. The two proofs follows in the same way as for linear codes.  The conclusion is then that $\sigma_j$ can be taken to be additive and we conclude that $\mathcal A$ and $\mathcal B$ are additively equivalent. Thus, we have the following theorem.
\begin{theorem} \label{addmdsthm}
Two additive MDS codes over a field are equivalent if and only if they are additively equivalent.
\end{theorem}

\begin{example}
In \cite[Table 3]{BGL2021}, the second row shows that there are 3 additive MDS $(8,9^3,6)_9$ codes, i.e. MDS codes of length $8$ and minimum distance $6$. Two of these codes, $C_1$ and $C_2$ are linear MDS codes and have generator matrices
$$
G_1=\left(\begin{array}{cccccccc}
1 & 0&  0 &  1 & e & e^6 & e^6 &  e \\
0 & 1 & 0 & e & 1 &  e & e^6 &  e^6 \\
0 & 0 & 1 & e^6 & e & 1 & e & e^6 \\
\end{array}\right), \ \ 
G_2=\left(\begin{array}{cccccccc}
1 & 0&  0 &  1 & e^5 & e^7 & e^7 &  e^5 \\
0 & 1 & 0 & e & 1 &  e^5   & e^7 &  e^7 \\
0 & 0 & 1 & e^7 & e^5 & 1 & e^5 & e^7 \\
\end{array}\right)
$$
respectively. Note that $e \in {\mathbb F}_9$ is a primitive element such that $e^2 = e + 1$. It is explained in \cite{BGL2021} that $C_2$ is a truncated Reed- Solomon code, while $C_1$ is not, implying that they are not semi-linearly equivalent. Additionally, one may show that $C_1$ and $C_2$ are not semi-linearly equivalent by confirming that the columns of $G_2$ are contained in the conic
$$
x_1x_2 + e^3x_1x_3 + x_2x_3,
$$
 while the columns of $G_1$ are not contained in any conic.
 
The last of the three additive MDS ($8,9^3,6)_9$ codes, $C_3$, is not semi-linearly equivalent to a linear code. It has generator matrix
$$
\left(
\begin{array}{cc|cc|cc|cc|cc|cc|cc|cc}
   1 & 0 & 0 & 0 & 0 & 0 & 1 & 0 & 1 & 0  & 1 & 0 & 1 & 0 & 1 & 0 \\
   0 & 1 & 0 & 0 & 0 & 0 & 0 & 1 & 0 & 1 & 0 & 1 & 0 & 1 & 0 & 1 \\
   0 & 0 & 1 & 0 & 0 & 0 & 1 & 0 & 2 & 2 & 0 & 2 & 2 & 0 & 1 & 2 \\
   0 & 0 & 0 & 1 & 0 & 0 & 0 & 1 & 0 & 2 & 1 & 1 & 1 & 2 & 1 & 0 \\
   0 & 0 & 0 & 0 & 1 & 0 & 1 & 0 & 0 & 1 & 2 & 0 & 1 & 1 & 2 & 1 \\
   0 & 0 & 0 & 0 & 0 & 1 & 0 & 1 & 2 & 1 & 2 & 2 & 2 & 0 & 0 & 2 \\
\end{array}\right).
$$
We now know, by Theorem~\ref{mainthm}, that $C_1$ and $C_2$ are not equivalent and, by Theorem~\ref{addmdsthm}, these codes are also not equivalent to $C_3$.

MDS codes have been classified (up to equivalence) over alphabets of size at most $8$, see \cite{Alderson2006}, \cite{KDO2015} and \cite{KO2016}. If and when MDS codes over alphabets of size $9$ are classified, there will be at least three equivalence classes of codes of length $8$, two of which will contain linear codes and one of which will contain an additive code. We conjecture that there will be no further equivalence class.
\end{example}

\section{Comments}

It is unclear to us if Theorem~\ref{addmdsthm} will apply to all additive codes. The principal problem in proving such a statement is that one cannot assume a generator matrix in standard form. It would be interesting to have a counterexample. In the proof of Theorem~\ref{addmdsthm}, we view additive MDS codes over ${\mathbb F}_q$, $q=p^h$, as linear codes over the vector space ${\mathbb F}_p^h$, i.e. $1\times h$ matrices over ${\mathbb F}_p$. It is  possible that one  can extend   Theorem~\ref{addmdsthm} to codes over general matrices. Indeed it would be interesting to study MDS codes over matrices. The smallest alphabet which is not equivalent to a linear or additive code over ${\mathbb F}_q$ would be over $2 \times  2$ matrices over ${\mathbb F}_2$; so the alphabet would have size $16$.

\vspace{1cm}

   Simeon Ball\\
   Departament de Matem\`atiques, \\
Universitat Polit\`ecnica de Catalunya, \\
M\`odul C3, Campus Nord,\\
Carrer Jordi Girona 1-3,\\
08034 Barcelona, Spain \\
   {\tt simeon.michael.ball@upc.edu} \\

James Dixon\\
Facultat de Matem\`atiques, \\
Universitat Polit\`ecnica de Catalunya, \\
Carrer de Pau Gargallo, 14, \\
08028 Barcelona, Spain\\
  {\tt james.dixon@estudiantat.upc.edu} \\

\end{document}